\documentclass[english,3p ]{llncs}
\usepackage[T1]{fontenc}
\usepackage[latin9]{inputenc}
\usepackage{babel}
\usepackage{units}
\usepackage{textcomp}
\usepackage{mathrsfs}
\usepackage{amsmath}
\usepackage{amssymb}
\usepackage[numbers]{natbib}
\usepackage[unicode=true]
 {hyperref}

\makeatletter

\newcommand{\lyxmathsym}[1]{\ifmmode\begingroup\def\b@ld{bold}
  \text{\ifx\math@version\b@ld\bfseries\fi#1}\endgroup\else#1\fi}

\usepackage[vlined,plain,noresetcount]{algorithm2e} 

\SetAlgoSkip{smallskip}
\usepackage{setspace}

\usepackage[adobe-utopia]{mathdesign}

\usepackage{todonotes}

\usepackage[margin=1in]{geometry}

\makeatother

\begin{document}

\title{Heuristic rating estimation - geometric approach}

\author{Konrad Ku\l{}akowski, Katarzyna Grobler-D\k{e}bska, Jaros\l{}aw W\k{a}s}

\institute{AGH University of Science and Technology, \\
al. Mickiewicza 30, Kraków, Poland, \\
\href{mailto:kkulak@agh.edu.pl}{kkulak@agh.edu.pl}, \href{mailto:grobler@agh.edu.pl}{grobler@agh.edu.pl},
\href{mailto:jarek@agh.edu.pl}{jarek@agh.edu.pl}}
\maketitle
\begin{abstract}
Heuristic Rating Estimation (HRE) is a newly proposed method supporting
decisions analysis based on the use of pairwise comparisons. It allows
that the ranking values of some alternatives (herein referred to as
concepts) are initially known, whilst the ranks for the other concepts
have yet to be estimated. To calculate the missing ranks it is assumed
that the priority of every single concept can be determined as the
weighted arithmetic mean of priorities of all the other concepts.
It has been shown that the problem has admissible solution if the
inconsistency of pairwise comparisons is not too high. 

The proposed approach adopts the heuristics according to which to
determine the missing priorities a weighted geometric mean is used.
In this approach, despite an increased complexity, the solution always
exists and their existence does not depend on the inconsistency of
the input matrix. Thus, the presented approach might be appropriate
for a larger number of problems than the previous method. The formal
definition of the proposed geometric heuristics is accompanied by
two numerical examples.
\end{abstract}

\section{Introduction }

The first written evidence about pairwise comparisons (PC) method
dates back to the thirteenth century, when \emph{Ramon Llull} from
Majorca wrote a seminal piece ``Artifitium electionis personarum''
(The method for the elections of persons) about voting and elections
\citep{Faliszewski2009lacv,Colomer2011rlfa}, followed by the two
consecutive works being a practical study on the election processes%
\footnote{\href{http://www.math.uni-augsburg.de/stochastik/llull/}{see: The Augsburg Web Edition of Llull's Electoral Writings}%
}. Nowadays PC as a voting method is a way of deciding on the relative
utility of alternatives used in decision theory \citep{Saaty1977asmfSIMPL}
and other fields like economy \citep{Peterson1998evabt}, psychometrics
and psychophysics \citep{Thurstone27aloc} and so on. The PC theory
is developed by many research teams representing different fields
and approaches. One can point out some characteristic approaches like
fuzzy PC relation developed by \emph{Kacprzyk} et al. and \emph{Mikhailov}
\citep{Kacprzyk2008ogdmSIMPLE,Mikhailov2003dpff}, data inconsistency
reduction methods proposed by \emph{Koczkodaj} and \emph{Szarek} \citep{Koczkodaj2010odbi}
and issue of incomplete PC relation by \emph{Koczkodaj} and \emph{Or\l{}owski}
\citep{Koczkodaj1999mnei} and \emph{Bozoki} and \emph{Rapcsak} \citep{Bozoki2010ooco},
problem of non-numerical rankings addressed by \emph{Janicki} and
\emph{Zhai} \citep{Janicki2012oapc} or using PC in Data Envelopment
Analysis \citep{Lotfi2011reui}. 

Currently, the Heuristic Rating Estimation (HRE) method which enables
the user to explicitly define the reference set of concepts, for which
the ranking values are a priori known, is being developed \citep{Kulakowski2013ahre,Kulakowski2013hrea}.
The base heuristics used in \emph{HRE} proposes to determine the relative
values of a single non--reference concept as a weighted arithmetic
mean of all the other concepts. This proposition leads to the linear
equation system defined by the matrix $A$ and the strictly positive
vector of constant terms $b$. 

In this work, the authors show that using a geometric mean to determine
the relative priorities of concepts instead of arithmetic one in some
cases may be more convenient. The main benefit of the proposed solution
stems from the guarantee of solution existence. Hence, unlike the
original proposal, the ranking list can always be created. This guarantee
is paid with the increase in computational complexity. The presented
solution is accompanied by two numerical examples. 

The presented work is a follow-up of research initiated in \citep{Kulakowski2013ahre,Kulakowski2013hrea}.
It redefines the main heuristics of HRE and the method of calculating
the solution. The HRE approach as proposed in the previous articles
is briefly outlined in (Sec. \ref{sec:Preliminaries}). There are
also a short summary of a few important properties of \emph{M-matrices}
(Sec. \ref{sub:M-matrices}), which are essential to the properties
of the presented method. The next section (Sec. \ref{sec:HRE-geometric})
describes the proposed solution and discusses two important properties:
solution existence (Sec. \ref{sub:Existence-of-solution}) and optimality
(Sec. \ref{sub:Optimality-condition}). Theoretical considerations
are accompanied by two meaningful examples showing how the presented
method can be used in practice (Sec \ref{sec:Numerical-examples}).
A brief summary is provided in (Sec.~\ref{sec:Summary}).

\section{Preliminaries\label{sec:Preliminaries} }

\subsection{Basic concepts of pairwise comparisons method}

The input to the \emph{PC} method is the\emph{ PC} matrix $M=(m_{ij})$,
where $m_{ij}\in\mathbb{R}_{+}$ and $i,j\in\{1,\ldots,n\}$. It expresses
a quantitative relation $R$ over the finite set of concepts $C\overset{\textit{df}}{=}\{c_{i}\in\mathscr{C}\,\text{and}\, i\in\{1,\ldots,n\}\}$
where $\mathscr{C}$ is a non empty universe of concepts, and $R(c_{i},c_{j})=m_{ij}$,
$R(c_{j},c_{i})=m_{ji}$. The values $m_{ij}$ and $m_{ji}$ represent
subjective expert judgment as to the relative importance, utility
or quality indicators of concepts $c_{i}$ and $c_{j}$. Thus, according
to the best knowledge of experts should holds that $c_{i}=m_{ij}c_{j}$
.
\begin{definition}
\label{def:A-matrix-recip}A matrix $M$ is said to be reciprocal
if for all $i,j\in\{1,\ldots,n\}$ holds $m_{ij}=\frac{1}{m_{ji}}$,
and $M$ is said to be consistent if for all $i,j,k\in\{1,\ldots,n\}$
is $m_{ij}\cdot m_{jk}\cdot m_{ki}=1$.
\end{definition}
Since the data in the \emph{PC} matrix represents subjective opinions
of experts, thus they might be inconsistent. Hence, it may exist a
triad $m_{ij},m_{jk},m_{ki}$ of entries in $M$ for which $m_{ik}\cdot m_{kj}\neq m_{ij}$.
This leads to the situation in which the relative importance of $c_{i}$
with respect to $c_{j}$ is either $m_{ik}\cdot m_{kj}$ or $m_{ij}$.
This observation underlies two related concepts: a priority deriving
method that transform even an inconsistent matrix $M$ into consistent
priority vector, and an inconsistency index describing how far the
matrix $M$ is inconsistent. There are a number of priority deriving
methods and inconsistency indexes \citep{Bozoki2008osak,Ishizaka2011rotm}.
For the purpose of the article the \emph{Koczkodaj's inconsistency
index} is adopted. 
\begin{definition}
\label{def:Koczkodaj's-inconsistency-index}Koczkodaj's inconsistency
index $\mathscr{K}$ of $n\times n$ and ($n>2)$ reciprocal matrix
$M$ is equal to 
\begin{equation}
\mathscr{K}(M)\overset{\textit{df}}{=}\underset{i,j,k\in\{1,\ldots,n\}}{\max}\left\{ \min\left\{ \left|1-\frac{m_{ij}}{m_{ik}m_{kj}}\right|,\left|1-\frac{m_{ik}m_{kj}}{m_{ij}}\right|\right\} \right\} \label{eq:koczkodaj_inc}
\end{equation}
 where $i,j,k=1,\ldots,n$ and $i\neq j\wedge j\neq k\wedge i\neq k$. 
\end{definition}
The result of the pairwise comparisons method is ranking - a function
that assigns values to the concepts. Formally, it can be defined as
follows. 
\begin{definition}
\label{def:ranking_fun}The ranking function for $C$ (the ranking
of $C$) is a function $\mu:C\rightarrow\mathbb{R}_{+}$ that assigns
to every concept from $C\subset\mathscr{C}$ a positive value from
$\mathbb{R}_{+}$. 
\end{definition}
Thus, $\mu(c)$ represents the ranking value for $c\in C$. The $\mu$
function is usually defined as a vector of weights $\mu\overset{\textit{df}}{=}\left[\mu(c_{1}),\ldots,\right.$
$\left.\mu(c_{n})\right]^{T}$. According to the most popular eigenvalue
based approach proposed by \emph{Saaty} \citep{Saaty1977asmfSIMPL}
the final ranking $\mu_{\textit{ev}}$ is determined as the principal
eigenvector of the $PC$ matrix $M$, rescaled so that the sum of
all its entries is $1$, i.e. 
\begin{equation}
\mu_{\textit{ev}}=\left[\frac{\mu_{\textit{max}}(c_{1})}{s_{\textit{ev}}},\ldots,\frac{\mu_{\textit{\textit{max}}}(c_{n})}{s_{\textit{ev}}}\right]^{T}\,\,\,\mbox{and}\,\,\, s_{\textit{ev}}=\underset{i=1}{\overset{n}{\sum}}\mu_{\textit{max}}(c_{i})\label{eq:eq:artihmetic_mean_meth}
\end{equation}

where $\mu_{\textit{ev}}$ - the ranking function, $\mu_{\textit{max}}\overset{\textit{df}}{=}\left[\mu_{\textit{max}}(c_{1}),\ldots,\right.$
$\left.\mu_{\textit{max}}(c_{n})\right]^{T}$ - the principal eigenvector
of $M$. Another popular approach proposes the rescaled geometric
mean (GM) of rows of $M$ as the ranking result, i.e. 
\begin{equation}
\mu_{gm}=\left[\frac{p_{1}}{s_{gm}},\ldots,\frac{p_{n}}{s_{\textit{gm}}}\right]^{T}\label{eq:geometric_mean_meth}
\end{equation}

where

\begin{equation}
p_{i}=\left(\prod_{j=1}^{n}m_{ij}\right)^{\frac{1}{n}}\,\,\,\,\,\,\text{and}\,\,\,\,\,\, s_{\textit{gm}}=\underset{i=1}{\overset{n}{\sum}}\left(\prod_{j=1}^{n}m_{ij}\right)^{\frac{1}{n}}\label{eq:geometric_mean_meth_sum}
\end{equation}

It can be shown that for the fully consistent matrix $M$ both ranking
vectors $\mu_{\textit{ev}}$ and $\mu_{\textit{gm}}$ are identical.
A more completely overview including other methods can be found in
\citep{Bozoki2008osak,Ishizaka2011rotm}.

\subsection{Pairwise comparisons method with the reference set}

Usually when using the pairwise comparisons method the ranking values
$\mu(c_{1}),\ldots,\mu(c_{n})$ are initially unknown. Hence they
are need to be determined by the priority deriving procedure. In some
cases, however, there are concepts for which the priorities are known
from elsewhere. Hence, the decision makers may have additional knowledge
about the group of elements $C_{K}\subseteq C$ that allow them to
determine $\mu(c)$ for $C_{K}$ in advance.

For example, let $c_{1},c_{2}$ and $c_{3}$ represent oil paintings
that an auction house plans to put for auction. The sequence of paintings
during the auction should correspond to their approximate valuation.
In order to determine the indicative price of paintings the auction
house asked experts to evaluate them in pairs taking into account
that two other paintings from the same period of time were previously
auctioned for $\mu(c_{4})$ and $\mu(c_{5})$. 

The situation as described above prompted the first author \citep{Kulakowski2013ahre,Kulakowski2013hrea}
to propose a \emph{Heuristic Rating Estimation (HRE)} model. According
to \emph{HRE} the set of concepts $C$ is composed of unknown concepts
$C_{U}=\{c_{1},\ldots,c_{k}\}$ and known (reference) concepts $C_{K}=\{c_{k+1},\ldots,c_{n}\}$,
where $C_{U},C_{K}\neq\emptyset$ and $C_{U}\cap C_{K}=\emptyset$.
The values $\mu(c_{i})$ for $c_{i}\in C_{K}$ are known, whilst the
values $\mu(c_{j})$ for elements $c_{j}\in C_{U}$ need to be calculated.
Following the heuristics of\emph{ averaging with respect to the reference
values} \citep{Kulakowski2013hrea} solution proposed by HRE is to
adopt as $\mu(c_{j})$, for every $c_{j}\in C_{U}$, the arithmetic
mean of all the other values $\mu(c_{i})$ multiplied by factor $m_{ji}$: 

\begin{equation}
\mu(c_{j})=\frac{1}{n-1}\sum_{i=1,i\neq j}^{n}m_{ji}\mu(c_{i})\label{eq:append3_eq1}
\end{equation}
If the experts judgments gathered in the matrix $M$ were fully consistent
(Def. \ref{def:A-matrix-recip}), then every component of the sum
(\ref{eq:append3_eq1}) in the form $m_{ji}\mu(c_{i})$ would equal
$\mu(c_{j})$. Because, it is generally not, then every component
is only an approximation of $\mu(c_{j})$. Thus, the arithmetic mean
of the individual approximations has been adopted as the most probable
value of $\mu(c_{j})$. To determine unknown values $\mu(c_{j})$
for $c_{j}\in C_{U}$ the problem formalised as (\ref{eq:append3_eq1})
can be written down as the linear equation system $A\mu=b$, where:

\begin{equation}
A=\left[\begin{array}{ccc}
1 & \cdots & -\frac{1}{n-1}m_{1,k}\\
-\frac{1}{n-1}m_{2,1} & \cdots & -\frac{1}{n-1}m_{2,k}\\
\vdots & \ddots & \vdots\\
-\frac{1}{n-1}m_{k,1} & \cdots & 1
\end{array}\right]\label{eq:A_matrix}
\end{equation}

and 

\begin{equation}
b=\left[\begin{array}{c}
\frac{1}{n-1}\sum_{i=k+1}^{n}m_{1,i}\mu(c_{i})\\
\frac{1}{n-1}\sum_{i=k+1}^{n}m_{2,i}\mu(c_{i})\\
\vdots\\
\frac{1}{n-1}\sum_{i=k+1}^{n}m_{k,i}\mu(c_{i})
\end{array}\right]\label{eq:b_vector}
\end{equation}

The solution $\mu=\left[\mu(c_{1}),\ldots,\mu(c_{k})\right]^{T}$
determines the values of $\mu$ for elements from $C_{U}$. Together
with known $\mu(c_{k+1}),\ldots,$ $\mu(c_{n})$ the vector $\mu$
forms the complete result list, which after sorting can be used to
build ranking. Although the values $\mu(c)$ for $c\in C$ are called
priorities, they usually have a specific meaning. In the case of previously
mentioned example they represent the expected price of paintings. 

According (Def. \ref{def:ranking_fun}) the ranking results must be
strictly positive, hence only strictly positive vectors $\mu$ are
considered as feasible. It can be shown that the equation $A\mu=b$
has a feasible solution if $A$ is strictly diagonally dominant by
rows \citep{Kulakowski2013hrea}. It has recently been shown that
the equation has a feasible solution when the inconsistency index
$\mathscr{K}(M)$ is not to high~\citep{Kulakowski2014note}.

\subsection{M-matrices\label{sub:M-matrices}}

Very often the real life problem can be reduced to the linear equation
system $A\mu=b$, where the matrix $A$ has some special structure.
Frequently the matrix $A$ has positive diagonal and nonpositive off-diagonal
entries. Due to their importance to the practice this type of matrix
was especially thoroughly studied by researchers \citep{Plemmons1976mcnm,Quarteroni2000nm}.
To define it formally a few more notions and definitions are needed. 

Let $\mathcal{M}_{\mathbb{R}}(n)$ be a set of $n\times n$ matrices
over $\mathbb{R}$, and $\mathcal{M}_{\mathbb{Z}}(n)$ the set of
all $A=[a_{ij}]\in\mathcal{M}_{\mathbb{R}}(n)$ with $a_{ij}\leq0$
if $i\neq j$ and $i,j\in\{1,\ldots,n\}$. Furthermore, assume that
for every matrix $A\in\mathcal{M}_{\mathbb{R}}(n)$ and vector $b\in\mathbb{R}^{n}$
the notation $A\geq0$ and $b\geq0$ will mean that every $m_{ij}$
and $b_{k}$ are non-negative and neither $A$ nor $b$ equals $0$.
The spectral radius of $A$ is defined as $\rho(A)\overset{\textit{df}}{=}\max\{|\lambda|:\det(\lambda I-A)=0\}$.
\begin{definition}
\label{def:M-matrix-def}An $n\times n$ matrix that can be expressed
in the form $A=sI-B$ where $B=[b_{ij}]$ with $b_{ij}\geq0$ for
$i,j\in\{1,\ldots,n\}$, and $s\geq\rho(B)$, the maximum of the moduli
of the eigenvalues of B, is called \emph{M-matrix}.
\end{definition}
Following \citep{Plemmons1976mcnm} some of the \emph{M-matrix} properties
are recalled below in the form of the Theorem~\ref{PlemmonsTheo}.
\begin{theorem}
\label{PlemmonsTheo}For every $A\in\mathcal{M}_{\mathbb{Z}}(n)$
each of the following conditions is equivalent to the statement: $A$
is a nonsingular \emph{M-matrix}.\end{theorem}
\begin{enumerate}
\item $A$ is inverse positive. That is, $A^{-1}$ exists and $A^{-1}\geq0$
\item There exists a positive diagonal matrix $D$ such that $AD$ has all
positive row sums.
\end{enumerate}
It is worth to note that for every matrix equation in the form $A\mu=b$,
where $A$ is a nonsingular M-matrix, holds $\mu=A^{-1}b$. Since
$A^{-1}\geq0$, thus, $b>0$ implies that also $\mu>0$.

\section{HRE - geometric approach\label{sec:HRE-geometric}}

\subsection{Heuristics of the geometric averaging with respect to the reference
values\label{sub:Heuristics-of-the}}

Most often the pairwise comparisons method is used to transform the
\emph{PC} matrix into the ranking list of mutually compared concepts.
During the transformation to each concept a priority is assigned.
Therefore, this transformation is often called a priority deriving
method. There are many priority deriving methods. Besides the eigenvalue
based method (\ref{eq:eq:artihmetic_mean_meth}), where the ranking
values $\mu(c_{i})$ are approximated as the arithmetic means of $m_{ij}\cdot\mu(c_{j})$,
also the geometric mean of rows is used (\ref{eq:geometric_mean_meth}).
This may suggest that also for the ranking problem with the reference
set \citep{Kulakowski2013hrea}, the arithmetic mean (\ref{eq:append3_eq1})
might be replaced by the geometric mean. This observation prompted
the author to formulate and investigate \emph{the geometric averaging
with respect to the reference values heuristics}. According to this
proposition to determine the unknown values $\mu(c_{j})$ for $c_{j}\in C_{U}$
the following non-linear equation is used: 
\begin{equation}
\mu(c_{j})=\left(\prod_{i=1,i\neq j}^{n}m_{ji}\mu(c_{i})\right)^{\frac{1}{n-1}}\label{eq:geometric_mean_proposal}
\end{equation}

After rising both sides to the $n-1$ power the geometric averaging
heuristics equation (\ref{eq:geometric_mean_proposal}) leads to the
non-linear equation system in the form: 
\begin{equation}
\begin{array}{ccc}
\mu^{n-1}(c_{1}) & = & m_{1,2}\mu(c_{2})\cdot\dotfill\cdot m_{1,n}\mu(c_{n})\\
\mu^{n-1}(c_{2}) & = & m_{2,1}\mu(c_{1})\cdot m_{2,3}\mu(c_{3})\cdot\ldots\cdot m_{2,n}\mu(c_{n})\\
\hdotsfor[1]{3}\\
\mu^{n-1}(c_{k}) & = & m_{k,1}\mu(c_{1})\cdot\dotfill\cdot m_{k,n-1}\mu(c_{n-1})
\end{array}\label{eq:non_linear_eq_system}
\end{equation}

Of course, since the ranking values for $c_{k+1},\ldots,c_{n}\in C_{K}$
make the reference set where the values $\mu(c_{j})$ are known and
fixed, some products in the form $m_{ji}\mu(c_{i})$ are initially
known constants. Let us denote:
\begin{equation}
g_{j}=\prod_{i=k+1}^{n}m_{ji}\mu(c_{i})\label{eq:b_j_constant}
\end{equation}

for $j=1,\ldots,k$ as the constant part of each equation (\ref{eq:non_linear_eq_system}).
Thus, the non-linear equation system can be written as: 
\[
\begin{array}{ccc}
\mu^{n-1}(c_{1}) & = & m_{1,2}\mu(c_{2})\cdot\dotfill\cdot m_{1,k}\mu(c_{k})\cdot g_{1}\\
\mu^{n-1}(c_{2}) & = & m_{2,1}\mu(c_{1})\cdot m_{2,3}\mu(c_{3})\cdot\ldots\cdot m_{2,k}\mu(c_{k})\cdot g_{2}\\
\hdotsfor[1]{3}\\
\mu^{n-1}(c_{k}) & = & m_{k,1}\mu(c_{1})\cdot\dotfill\cdot m_{k,k-1}\mu(c_{k-1})\cdot g_{k}
\end{array}
\]

Hence $\mu(c_{j}),\, m_{ij},\, g_{j}\in\mathbb{R}_{+}$, let us denote
$\log_{\xi}\mu(c_{j})\overset{\textit{df}}{=}\widehat{\mu}(c_{j})$,
$\widehat{m}_{ij}\overset{\textit{df}}{=}\log_{\xi}m_{ij}$ and $\widehat{g}_{j}\overset{\textit{df}}{=}\log_{\xi}g_{j}$
for some $\xi\in\mathbb{R}_{+}$. It is easy to see that the above
non-linear equation system is equivalent to the following one: 
\begin{equation}
\begin{array}{ccc}
(n-1)\widehat{\mu}(c_{1}) & = & \widehat{m}_{1,2}+\widehat{\mu}(c_{2})+\dotfill+\widehat{m}_{1,k}+\widehat{\mu}(c_{k})+\widehat{g}_{1}\\
(n-1)\widehat{\mu}(c_{2}) & = & \widehat{m}_{2,1}+\widehat{\mu}(c_{1})+\dotfill+\widehat{m}_{2,k}+\widehat{\mu}(c_{k})+\widehat{g}_{2}\\
\hdotsfor[1]{3}\\
(n-1)\widehat{\mu}(c_{k}) & = & \widehat{m}_{k,1}+\widehat{\mu}(c_{1})+\ldots+\widehat{m}_{k,k-1}+\widehat{\mu}(c_{k-1})+\widehat{g}_{k}
\end{array}\label{eq:linear_equation_system_1}
\end{equation}

By grouping all the constant terms on the right side of each above
equation we obtain the linear equation system 

\begin{equation}
\begin{array}{ccc}
(n-1)\widehat{\mu}(c_{1})-\sum_{i=2}^{k}\widehat{\mu}(c_{i})\,\,\,\,\,\,\,\,\,\,\, & = & b_{1}\\
(n-1)\widehat{\mu}(c_{2})-\sum_{i=1,i\neq2}^{k}\widehat{\mu}(c_{i}) & = & b_{2}\\
\hdotsfor[1]{3}\\
(n-1)\widehat{\mu}(c_{k})-\sum_{i=1}^{k-1}\widehat{\mu}(c_{i})\,\,\,\,\,\,\,\,\,\, & = & b_{k}
\end{array}\label{eq:linear_equation_system_2}
\end{equation}

where $b_{i}\overset{\textit{df}}{=}\sum_{j=1,j\neq i}^{k}\widehat{m}_{1,j}+\widehat{g}_{i}$
for $i=1,\ldots,k$, which can be easily written down in the matrix
form 
\begin{equation}
\widehat{A}\widehat{\mu}=b\label{eq:linear_equation_system_2a}
\end{equation}
where: 

\begin{equation}
\widehat{A}=\left[\begin{array}{cccc}
(n-1) & -1 & \cdots & -1\\
\vdots & \ddots &  & \vdots\\
\vdots &  & \ddots & \vdots\\
-1 & -1 & \cdots & (n-1)
\end{array}\right],\label{eq:linear_equation_system_3}
\end{equation}

\begin{equation}
\widehat{\mu}=\left[\begin{array}{c}
\widehat{\mu}(c_{1})\\
\widehat{\mu}(c_{2})\\
\vdots\\
\widehat{\mu}(c_{k})
\end{array}\right],\,\,\,\text{and}\,\,\, b=\left[\begin{array}{c}
b_{1}\\
b_{2}\\
\vdots\\
b_{k}
\end{array}\right]\label{eq:linear_equation_system_3b}
\end{equation}

Therefore, the solution $\widehat{\mu}$ of the linear equation system
(\ref{eq:linear_equation_system_2a}) automatically provides the solution
to the original non-linear problem as formulated in (\ref{eq:non_linear_eq_system}).
Indeed the ranking vector $\mu$ can be computed following the formula:
\begin{equation}
\mu=\left[\xi^{\widehat{\mu}(c_{1})},\ldots,\xi^{\widehat{\mu}(c_{k})}\right]^{T}\label{eq:geometric_mean_solution}
\end{equation}

Importantly, as it is shown below a feasible solution of (\ref{eq:linear_equation_system_2a})
always exists. Hence, the heuristics of the averaging with respect
to the geometric mean always provides the user an appropriate ranking
function.

\subsection{Existence of solution\label{sub:Existence-of-solution}}

The form of $\widehat{A}$ is specific. The positive diagonal and
the negative off-diagonal real entries cause that $\widehat{A}\in\mathcal{M}_{\mathbb{Z}}(k)$
(see Sec. \ref{sub:M-matrices}). Let us put:
\[
D=\left[\begin{array}{ccc}
1 & \cdots & 0\\
\vdots & \ddots & \vdots\\
0 & \cdots & 1
\end{array}\right]
\]
and $D\in\mathcal{M}_{\mathbb{R}}(k)$. Of course $D$ is positively
dominant matrix. Thus, the product $\widehat{A}\cdot D=\widehat{A}$.
The sum of each row in $\widehat{A}$ equals 
\[
(n-1)+\sum_{i=1}^{k-1}(-1)=n-k
\]
Since $C_{K}$ is nonempty, thus its cardinality $\left|C_{K}\right|=n-k$
is greater than $0$. This means that the sum of each row of $\widehat{A}\cdot D$
is positive. Hence, due to the Theorem \ref{PlemmonsTheo}, $\widehat{A}$
is a nonsingular M-matrix (Def. \ref{def:M-matrix-def}). Thus, $\widehat{A}^{-1}$
exists (i.e. $\widehat{\mu}=\widehat{A}^{-1}b$) and always the equation
(\ref{eq:linear_equation_system_2a}) has a solution in $\mathbb{R}^{k}$.
Due to the form of the solution of the main problem (\ref{eq:geometric_mean_solution})
$\mu$ is a vector in $\mathbb{R}_{+}^{k}$, i.e. every its entry
is strictly positive. In other words unlike the original proposition
\citep{Kulakowski2013hrea} the heuristics of the geometric averaging
with respect to the reference values always provides a feasible ranking
result to the user.

\subsection{Optimality condition\label{sub:Optimality-condition}}

One of the reasons for introducing the geometric mean method (\ref{eq:geometric_mean_meth})
is minimizing the multiplicative error $e_{ij}$ \citep{Ishizaka2011rotm}
defined as:
\begin{equation}
m_{ij}=\frac{p_{i}}{p_{j}}e_{ij}\label{eq:error_cond_orig}
\end{equation}
In the case of the geometric averaging heuristics the multiplicative
error equation takes the form: 

\begin{equation}
m_{ij}=\frac{\mu(c_{i})}{\mu(c_{j})}e_{ij}\label{eq:error_cond_modif}
\end{equation}
The multiplicative error is commonly accepted to be log normal distributed
(in the same way the additive error would be assumed to be normally
distributed). Let $e:\mathbb{R}_{+}^{n}\rightarrow\mathbb{R}$ be
the sum of multiplicative errors (see \citep{Ishizaka2011rotm}) defined
as follow: 

\begin{equation}
e(\mu(c_{1}),\ldots,\mu(c_{n}))=\sum_{i=1}^{n}\sum_{j=1}^{n}\left(\ln(m_{ij})-\ln\left(\frac{\mu(c_{i})}{\mu(c_{j})}\right)\right)^{2}\label{eq:error_funct}
\end{equation}
As it is shown in the Theorem below very often the heuristics (\ref{eq:geometric_mean_proposal})
is optimal with respect to the value of multiplicative error function
$e$. 
\begin{theorem}
The geometric averaging with respect to the reference values heuristics
minimizes the sum of multiplicative errors $e(\mu(c_{1}),\ldots,\mu(c_{n}))$
if 
\begin{equation}
\mu(c_{i})<(n-1)\sum_{j=1,j\neq i}^{n}\mu(c_{j})\label{eq:optimality_condition_theorem}
\end{equation}

for $i=1,\ldots,n$. \end{theorem}
\begin{proof}
To determine the minimum of (\ref{eq:error_funct}) let us forget
for a moment that $\mu(c_{k+1}),\ldots,\mu(c_{n})$ are constants
(the reference values), and let us treat them as any other arguments
of $e$. In order to determine the minimum of (\ref{eq:error_funct})
the first derivative need to be calculated. Thus, 
\end{proof}
\begin{align}
\frac{\partial e}{\partial\mu(c_{i})} & =\frac{1}{\mu(c_{i})}\left(\sum_{r=1,r\neq i}^{n}4(n-1)\ln\mu(c_{i})-4\sum_{j=1,j\neq i}^{n}\ln\mu(c_{j})+2\sum_{r=1,r\neq i}^{n}\ln(m_{ri})-2\sum_{j=1,j\neq i}^{n}\ln(m_{ij})\right)\label{eq:first_derivative_1}
\end{align}
for $i=1,\ldots,n$. Due to the reciprocity of $M$, i.e. $m_{ij}=\nicefrac{1}{m_{ji}}$,
the equation (\ref{eq:first_derivative_1}) can be written as:
\begin{align}
\frac{\partial e}{\partial\mu(c_{i})} & =-4\left(\frac{\sum_{j=1,j\neq i}^{n}(\ln\mu(c_{j})+\ln(m_{ij}))-(n-1)\ln\mu(c_{i})}{\mu(c_{i})}\right)\label{eq:first_derivative_2}
\end{align}

\[
\]
The function $e$ reaches the minimum if $\nicefrac{\partial e}{\partial\mu(c_{i})}=0$.
This leads to the postulate that 
\begin{equation}
\sum_{j=1,j\neq i}^{n}(\ln\mu(c_{j})+\ln(m_{ij}))-(n-1)\ln\mu(c_{i})=0\label{eq:zeroes_postulate}
\end{equation}
for $i=1,\ldots,n$. Thus, 
\begin{equation}
\ln\mu(c_{i})=\frac{1}{n-1}\left(\sum_{j=1,j\neq i}^{n}\ln m_{ij}\mu(c_{j})\right)\label{eq:almost_heuresis}
\end{equation}
which is directly equivalent to (\ref{eq:geometric_mean_proposal}).
In other words any solution to the equation system (\ref{eq:non_linear_eq_system})
is a good candidate to be a minimum of (\ref{eq:error_funct}). It
remains to settle the matrix $H$ of second derivative of $e$. When
$H$ is positive definite then the solution of (\ref{eq:non_linear_eq_system})
actually minimizes the function $e$. As a result of further differentiation
is determined that the diagonal elements of $H$ are 
\begin{equation}
\frac{\partial^{2}f}{\partial\mu(c_{i})\partial\mu(c_{i})}=\frac{4(n-1)}{\mu^{2}(c_{i})}-\frac{1}{\mu(c_{i})}\frac{\partial f}{\partial\mu(c_{i})}\label{eq:diagonal_hessian_element}
\end{equation}

\begin{proof}
where $i=1,\ldots,n$, and the other elements for which $i\neq j$
and $i,j=1,\ldots,n$ take the form:
\begin{equation}
\frac{\partial^{2}f}{\partial\mu(c_{i})\partial\mu(c_{j})}=-\frac{4}{\mu(c_{i})\mu(c_{j})}\label{eq:off-diagonal-hessian-element}
\end{equation}
Since the matrix $H$ is considered for $e$ in the point $\left(\mu(c_{1}),\ldots,\right.$
$\left.\mu(c_{n})\right)$ such that (\ref{eq:geometric_mean_proposal})
holds, thus the first derivative of $e$ is $0$. Therefore, the Hessian
matrix $H$ takes the form:
\begin{equation}
H=\left[\begin{array}{cccc}
\frac{4(n-1)}{\mu^{2}(c_{1})} & -\frac{4}{\mu(c_{1})\mu(c_{2})} & \cdots & -\frac{4}{\mu(c_{1})\mu(c_{n})}\\
\vdots & \frac{4(n-1)}{\mu^{2}(c_{2})} & \vdots & \vdots\\
\vdots & \vdots & \ddots & \vdots\\
-\frac{4}{\mu(c_{n})\mu(c_{1})} & -\frac{4}{\mu(c_{n})\mu(c_{2})} & \cdots & \frac{4(n-1)}{\mu^{2}(c_{n})}
\end{array}\right]\label{eq:hessian_matrix_in_solution}
\end{equation}
According to \citep[p. 29]{Quarteroni2000nm} if $H$ is strictly
diagonally dominant by rows, symmetric, and with positive diagonal
entries then it is also positive definite. To meet the first strict
diagonal dominance criterion (other are satisfied) it is required
that:
\begin{equation}
\left|\frac{n-1}{\mu^{2}(c_{i})}\right|>\sum_{j=1,j\neq i}^{n}\left|-\frac{1}{\mu(c_{i})\mu(c_{j})}\right|\label{eq:strict_dominance_postulate}
\end{equation}

for $i=1,\ldots,n$. Thus, 
\begin{equation}
\mu^{2}(c_{i})<(n-1)\mu(c_{i})\sum_{j=1,j\neq i}^{n}\mu(c_{j})\label{eq:strict_dominance_postulate_2}
\end{equation}

Since every $\mu(c_{i})>0$, then it is easy to verify that the above
equation is equivalent to the desired condition (\ref{eq:optimality_condition_theorem}). \end{proof}

\section{Numerical examples\label{sec:Numerical-examples}}

The HRE method can be useful in many situations in which, based on
the expert subjective opinions and the actual data, the new concepts,
objects or entities need to be assessed. In order to show how the
method may work in practice the following two numerical examples are
presented. The first one, more abstract, discusses the method for
solving the non-linear equation system. The second one, more complex,
tries to put the method into the actual business context, where it
can be successfully used. 

In both examples the set of concepts consists of $C_{K}$ - the reference
(known) and $C_{U}$ - the initially unknown elements. To solve an
intermediate linear equation system (\ref{eq:linear_equation_system_2a})
the Gaussian elimination method is used.

\subsection{Example I (Scientific entities assessment)}

Let $c_{1},...,c_{5}$ represent the scientific entities%
\footnote{Actually the official ranking of the scientific entities in Poland
compares the entities in pairs \citep{Koczkodaj2014otqe}.%
}, where two of them $c_{2},c_{3}\in C_{K}$ are the reference entities.
Their values were arbitrarily set by experts to $\lyxmathsym{\textmu}(c_{2})=5$
and $\lyxmathsym{\textmu}(c_{3})=7$. The analysis of the scientific
achievements of the entities $c_{1},c_{4}$ and $c_{5}$ leads to
the following PC matrix: 
\begin{equation}
M=\left[\begin{array}{ccccc}
1 & \frac{3}{5} & \frac{4}{7} & \frac{5}{8} & \frac{5}{9}\\
\frac{5}{3} & 1 & \frac{5}{7} & \frac{5}{2} & \frac{10}{3}\\
\frac{7}{4} & \frac{7}{5} & 1 & \frac{7}{2} & 4\\
\frac{8}{5} & \frac{2}{5} & \frac{2}{7} & 1 & \frac{4}{3}\\
\frac{9}{5} & \frac{3}{10} & \frac{1}{4} & \frac{3}{4} & 1
\end{array}\right]
\end{equation}

To calculate the rank using HRE with the geometric averaging heuristics,
the following system of non-linear equations (compare with \ref{eq:non_linear_eq_system})
need to be solved:

\begin{equation}
\begin{array}{ccc}
\mu(c_{1}) & = & \left(m_{1,2}\mu(c_{2})\right.\cdot\dotfill\cdot\left.m_{1,5}\mu(c_{5})\right)^{\frac{1}{4}}\\
\mu(c_{4}) & = & \left(m_{4,1}\mu(c_{1})\right.\cdot\ldots\cdot m_{4,3}\mu(c_{1})\cdot\left.m_{4,5}\mu(c_{5})\right)^{\frac{1}{4}}\\
\mu(c_{5}) & = & \left(m_{5,1}\mu(c_{1})\right.\cdot\dotfill\cdot\left.m_{5,4}\mu(c_{4})\right)^{\frac{1}{4}}
\end{array}\label{eq:ex5}
\end{equation}

thus, after rising both sides of the equations to the power, 
\begin{equation}
\begin{array}{ccc}
\mu^{4}(c_{1}) & = & m_{1,2}\mu(c_{2})\cdot\dotfill\cdot m_{1,5}\mu(c_{5})\\
\mu^{4}(c_{4}) & = & m_{4,1}\mu(c_{1})\cdot\ldots\cdot m_{4,3}\mu(c_{1})\cdot m_{4,5}\mu(c_{5})\\
\mu^{4}(c_{5}) & = & m_{5,1}\mu(c_{1})\cdot\dotfill\cdot m_{5,4}\mu(c_{4})
\end{array}\label{eq:ex6}
\end{equation}

Substituting the logarithm of both sides of the equations, we get
the following system:

\begin{equation}
\begin{array}{ccc}
4\lg\mu(c_{1}) & = & \lg\left(m_{1,2}\mu(c_{2})\right.\cdot\dotfill\cdot\left.m_{1,5}\mu(c_{5})\right)\\
4\lg\mu(c_{4}) & = & \lg\left(m_{4,1}\mu(c_{1})\right.\cdot\ldots\cdot m_{4,3}\mu(c_{1})\cdot\left.m_{4,5}\mu(c_{5})\right)\\
4\lg\mu(c_{5}) & = & \lg\left(m_{5,1}\mu(c_{1})\right.\cdot\dotfill\cdot\left.m_{5,4}\mu(c_{4})\right)
\end{array}\label{eq:ex7}
\end{equation}

which leads to the intermediate, linear logarithmic equation system:

\begin{equation}
\begin{array}{ccc}
4\lg\mu(c_{1})-\lg\mu(c_{4})-\lg\mu(c_{5}) & = & b_{1}\\
-\lg\mu(c_{1})+4\lg\mu(c_{4})-\lg\mu(c_{5}) & = & b_{4}\\
-\lg\mu(c_{1})-\lg\mu(c_{4})+4\lg\mu(c_{5}) & = & b_{5}
\end{array}\label{eq:ex9}
\end{equation}

where

\begin{equation}
\begin{array}{ccc}
b_{1} & \overset{\textit{df}}{=} & \lg\left(m_{1,2}\mu(c_{2})m_{1,3}\mu(c_{3})m_{1,4}m_{1,5}\right)\\
b_{4} & \overset{\textit{df}}{=} & \lg\left(m_{4,1}m_{4,2}\mu(c_{2})m_{4,3}\mu(c_{3})m_{4,5}\right)\\
b_{5} & \overset{\textit{df}}{=} & \lg\left(m_{5,1}m_{5,2}\mu(c_{2})m_{5,3}\mu(c_{3})m_{5,4}\right)
\end{array}\label{eq:ex10}
\end{equation}

Then, according to the procedure proposed in (Sec. \ref{sub:Heuristics-of-the})
the linear equation system (\ref{eq:linear_equation_system_2a}) where
the unknown values $\hat{\lyxmathsym{\textmu}}(c_{i})\overset{\textit{df}}{=}\lg\left(\lyxmathsym{\textmu}(c_{i})\right)$
for $i=1,4,5$ takes the form: 
\begin{equation}
\left[\begin{array}{ccc}
n-1 & -1 & -1\\
-1 & n-1 & -1\\
-1 & -1 & n-1
\end{array}\right]\left[\begin{array}{c}
\widehat{\mu}(c_{1})\\
\widehat{\mu}(c_{4})\\
\widehat{\mu}(c_{5})
\end{array}\right]=\left[\begin{array}{c}
b_{1}\\
b_{4}\\
b_{5}
\end{array}\right]
\end{equation}

hence, numerically:
\begin{equation}
\left[\begin{array}{ccc}
4 & -1 & -1\\
-1 & 4 & -1\\
-1 & -1 & 4
\end{array}\right]\left[\begin{array}{c}
\widehat{\mu}(c_{1})\\
\widehat{\mu}(c_{4})\\
\widehat{\mu}(c_{5})
\end{array}\right]=\left[\begin{array}{c}
0.62\\
0.949\\
0.537
\end{array}\right]
\end{equation}

Solving the linear equation system provides us with $\widehat{\mu}(c_{1})=0.335$,
$\widehat{\mu}(c_{4})=0.4$ and $\widehat{\mu}(c_{5})=0.318$ which
leads to the desired result $10^{\widehat{\mu}(c_{1})}=2.16$, $10^{\widehat{\mu}(c_{4})}=2.514$
and $10^{\widehat{\mu}(c_{5})}=2.08$. The non-scaled weight vector
$\mu$ supplemented by the known values $\mu(c_{2})=5$ and $\mu(c_{3})=7$
takes the form:

\begin{equation}
\mu=\left[2.16,5,7,2.514,2.08\right]^{T}
\end{equation}

and after rescaling:

\begin{equation}
\mu_{n}=\left[0.115,0.267,0.373,0.134,0.111\right]^{T}
\end{equation}

Note that $|C_{U}|=3$ implies that the dimensions of matrix $\hat{A}$
are $3\times3$, moreover $\det(\hat{A})\neq0$ and $\lyxmathsym{\textmu}(c_{i})>0$
for $i=1,4,5$ (see sec. \ref{sub:Existence-of-solution}) .

\subsection{Example II (Choosing the best TV show)}

Certain TV broadcaster wants to produce a new entertainment TV show
in one of the European countries. It considering a purchase the license
for one of the five entertainment shows produced in the United States.
 So far in Europe three similar programs were broadcasted. Through
the market research there are known approximate size of their European
audience. They are respectively $5,500,000,\,4,500,000$ and $4,950,000$
persons for programs $c_{6},c_{7}$ and $c_{8}$ correspondingly.
The production costs of these programs are similar. In order to select
possibly the most profitable TV show the station hires a few seasoned
media experts. During the expert panel they prepared the following
\emph{PC} matrix $M$ representing a relative attractiveness of all
the considered programs. 

{\renewcommand{\arraycolsep}{1.3pt}

\begin{equation}
M=\left[\begin{array}{cccccccc}
1 & 0.8 & 1.333 & 0.7 & 0.5 & 0.6 & 0.75 & 0.667\\
1.25 & 1 & 1.667 & 0.875 & 0.625 & 0.75 & 0.9 & 0.833\\
1.333 & 0.6 & 1 & 0.933 & 0.667 & 0.8 & 0.978 & 0.889\\
1.429 & 1.143 & 1.071 & 1 & 0.714 & 0.857 & 1.05 & 0.952\\
2 & 1.6 & 1.5 & 1.4 & 1 & 1.2 & 1.467 & 1.333\\
1.667 & 1.333 & 1.25 & 1.167 & 0.833 & 1 & 1.222 & 1.111\\
1.333 & 1.111 & 1.023 & 0.952 & 0.682 & 0.818 & 1 & 0.909\\
1.5 & 1.2 & 0.382 & 1.05 & 0.75 & 0.9 & 1.1 & 1
\end{array}\right]
\end{equation}

}

In the matrix $M$ every entry $m_{ij}$ corresponds to the ratio
describing attractiveness of the \emph{TV} show $c_{i}$ with respect
to the attractiveness of \emph{TV} show $c_{j}$. Since the values
of attractiveness for $c_{6},c_{7}$ and $c_{8}$ are known (they
are approximated by the number of people watching the given TV show),
thus the appropriate ratios $m_{ij}$ for $i,j=6,7,8$ are not the
subject of the expert judgment. Instead, they are calculated based
on data from the market research. For example: 
\begin{equation}
m_{6,7}=\frac{\mu(c_{6})}{\mu(c_{7})}=\frac{5,100,000}{4,500,000}=1.222\label{eq:sample_m_67}
\end{equation}

or 

\begin{equation}
m_{6,8}=\frac{\mu(c_{6})}{\mu(c_{8})}=\frac{5,100,000}{4,950,000}=1.111\label{eq:sample_m_68}
\end{equation}
The other entries of $M$ represent the subjective judgements of experts. 

Similarly as before, to find a solution with the help of HRE supported
by the geometric averaging heuristics, the system of equations (\ref{eq:non_linear_eq_system})
must be solved. The desired values $\lyxmathsym{\textmu}(c_{i})$
for $i=1\ldots,5$ will be derived from the formula $\widehat{\mu}(c_{i})=\log\lyxmathsym{\textmu}(c_{i})$.
Because $|C_{U}|=5$, the dimensions of matrix $\widehat{A}$ are
$5\times5$. The linear equation system need to be solved is as follows:

{\renewcommand{\arraycolsep}{1.3pt}

\begin{equation}
\left[\begin{array}{ccccc}
n-1 & -1 & -1 & -1 & -1\\
-1 & n-1 & -1 & -1 & -1\\
-1 & -1 & n-1 & -1 & -1\\
-1 & -1 & -1 & n-1 & -1\\
-1 & -1 & -1 & -1 & n-1
\end{array}\right]\left[\begin{array}{c}
\begin{array}{c}
\widehat{\mu}(c_{1})\\
\widehat{\mu}(c_{2})\\
\widehat{\mu}(c_{3})
\end{array}\\
\begin{array}{c}
\widehat{\mu}(c_{4})\\
\widehat{\mu}(c_{5})
\end{array}
\end{array}\right]=\left[\begin{array}{c}
b_{1}\\
b_{2}\\
b_{3}\\
b_{4}\\
b_{5}
\end{array}\right]\label{eq:intermediate_eq_ex2}
\end{equation}

where 
\begin{equation}
\begin{array}{ccc}
b_{1} & \overset{\textit{df}}{=} & \lg\left(m_{1,2}m_{1,3}m_{1,4}m_{1,5}m_{1,6}\mu(c_{6})m_{1,7}\mu(c_{7})m_{1,8}\mu(c_{8})\right)\\
b_{2} & \overset{\textit{df}}{=} & \lg\left(m_{2,1}m_{2,3}m_{2,4}m_{2,5}m_{2,6}\mu(c_{6})m_{2,7}\mu(c_{7})m_{2,8}\mu(c_{8})\right)\\
b_{3} & \overset{\textit{df}}{=} & \lg\left(m_{3,1}m_{3,2}m_{3,4}m_{3,5}m_{3,6}\mu(c_{6})m_{3,7}\mu(c_{7})m_{3,8}\mu(c_{8})\right)\\
b_{4} & \overset{\textit{df}}{=} & \lg\left(m_{4,1}m_{4,2}m_{4,3}m_{4,5}m_{4,6}\mu(c_{6})m_{4,7}\mu(c_{7})m_{4,8}\mu(c_{8})\right)\\
b_{5} & \overset{\textit{df}}{=} & \lg\left(m_{5,1}m_{5,2}m_{5,3}m_{5,4}m_{5,6}\mu(c_{6})m_{5,7}\mu(c_{7})m_{5,8}\mu(c_{8})\right)
\end{array}\label{eq:b_values_ex2}
\end{equation}

hence, (\ref{eq:intermediate_eq_ex2}) numerically:
\begin{equation}
\left[\begin{array}{ccccc}
7 & -1 & -1 & -1 & -1\\
-1 & 7 & -1 & -1 & -1\\
-1 & -1 & 7 & -1 & -1\\
-1 & -1 & -1 & 7 & -1\\
-1 & -1 & -1 & -1 & 7
\end{array}\right]\left[\begin{array}{c}
\begin{array}{c}
\widehat{\mu}(c_{1})\\
\widehat{\mu}(c_{2})\\
\widehat{\mu}(c_{3})
\end{array}\\
\begin{array}{c}
\widehat{\mu}(c_{4})\\
\widehat{\mu}(c_{5})
\end{array}
\end{array}\right]=\left[\begin{array}{c}
\begin{array}{c}
19.137\\
19.895
\end{array}\\
19.627\\
\begin{array}{c}
20.118\\
21.286
\end{array}
\end{array}\right]\label{eq:intermediate_eq_num_ex2}
\end{equation}

}

The intermediate result vector is:

\begin{equation}
\widehat{\mu}=\left[6.561,6.656,6.623,6.684,6.83\right]^{T}\label{eq:intermediate_result_ex2}
\end{equation}

Hence, following the rule $\mu(c_{i})=\xi^{\widehat{\mu}(c_{i})}$,
where $\xi=10$ is the logarithm base, the final result vector is
calculated. 

\begin{equation}
\mu=\left[\begin{array}{c}
3,643,307\\
4,530,955\\
4,196,128\\
4,831,326\\
6,761,938
\end{array}\right]\label{eq:ex_last}
\end{equation}

Thus, according to the expert judgments and the market research the
TV show number $5$ (denoted as $c_{5}$) has a chance to gather in
front of TVs near $6.8$ million people, whilst the second one in
line ``only'' $4.8$ million of people. Based on this estimate the
board of directors representing the broadcaster has decide to recommend
the purchase of the license for the fifth presented TV show.

\section{Summary \label{sec:Summary}}

The presented geometric HRE approach is another solution to the problem
of rankings with the reference set. It proposes to use a geometric
mean instead of arithmetic one used in \citep{Kulakowski2013ahre,Kulakowski2013hrea}.
The advantage of this approach is the robustness of the procedure.
As has been shown in (Sec. \ref{sub:Existence-of-solution}) the proposed
solution works for arbitrary set of input data producing admissible
vector of weights. The resulted ranking very often turns out to be
optimal in sense of the magnitude of multiplicative errors. According
to the formulated and proven condition (Sec. \ref{sub:Optimality-condition}),
this happens when the differences between the resulted priorities
are not too large. 

The \emph{HRE} approach may be useful in many different situations
including, ranking creation, valuation of goods and services, risk
assessment and others. Due to the lack of restrictions on the input
\emph{PC} matrix (method with the geometric mean always produces an
admissible result), the scope of the applicability of the HRE method
increases. Thus, the presented method covers cases which can not always
be dealt with using the arithmetic mean heuristics.

Despite the encouraging results, much remains to be done. In particular,
the role of the inconsistency in the input matrix $M$ should be more
deeply investigated.  Of course, the more studied examples, the better.
Thus, further development of the method will be particularly focused
on the study and analysis of use cases.

\bibliographystyle{plain}
\bibliography{papers_biblio_reviewed}

\end{document}